\documentclass[runningheads]{llncs}
\usepackage[T1]{fontenc}
\usepackage{graphicx}  %
\usepackage{amsmath,amssymb,amsfonts}
\usepackage[ruled,linesnumbered, boxed]{algorithm2e}  %
\usepackage{hyperref, xcolor}  %
\hypersetup{colorlinks=true, linkcolor=blue, citecolor=blue}
\usepackage[misc]{ifsym} %
\usepackage{subfigure}
\usepackage{float}
\usepackage{threeparttable}  %
\usepackage{booktabs}  %
\usepackage{cite, bm}
\usepackage{orcidlink}  %
\usepackage{hyperref}
\hypersetup{colorlinks=true, linkcolor=blue, citecolor=blue}
\usepackage{color}

\newcommand{\etal}{\textit{et al}. }
\newcommand{\ie}{\textit{i}.\textit{e}., }
\newcommand{\eg}{\textit{e}.\textit{g}., }

\newcommand{\eqnref}[1]{Eq.~(\ref{#1})}  %
\newcommand{\tabincell}[2]{\begin{tabular}{@{}#1@{}}#2\end{tabular}}  %
\usepackage[ruled,linesnumbered, boxed]{algorithm2e}  %
\SetKwFunction{decrypt}{decrypt}
\SetKwFunction{encrypt}{encrypt}
\SetKwFunction{hash}{hash}
\SetKwFunction{Fverify}{$\mathcal{V}$}
\SetKwFunction{Fsigma}{$\sigma$}
\makeatletter
\def\thanks#1{\protected@xdef\@thanks{\@thanks
        \protect\footnotetext{#1}}}
\makeatother
\usepackage[hang,flushmargin]{footmisc}
\begin{document}
\title{RecAGT: Shard Testable Codes with Adaptive Group Testing for Malicious Nodes Identification in Sharding Permissioned Blockchain}
\titlerunning{RecAGT: Shard Testable Codes with Adaptive Group Testing}
\author{Dong-Yang Yu\inst{1}\orcidlink{0009-0003-3830-9388}
\and Jin Wang\inst{1,2}$^{(\textrm{\Letter})}$\orcidlink{0000-0003-0766-9906}
\and Lingzhi Li\inst{1}$^{(\textrm{\Letter})}$\orcidlink{0000-0003-3336-2369}
\and Wei Jiang\inst{1}
\and Can Liu\inst{1}
}
\authorrunning{D.Y. Yu et al.}
\institute{
School of Computer Science and Technology, Soochow University, Suzhou, China 
\and School of Future Science and Engineering, Soochow University, Suzhou, China
\email{dyyu@stu.suda.edu.cn}\\
}
\thanks{Corresponding Authors: Jin Wang and Lingzhi Li}
\maketitle %

\begin{abstract}
Recently, permissioned blockchain has been extensively explored in various fields, such as asset management, supply chain, healthcare, and many others. Many scholars are dedicated to improving its verifiability, scalability, and performance based on sharding techniques, including grouping nodes and handling cross-shard transactions.
However, they ignore the node vulnerability problem, \ie there is no guarantee that nodes will not be maliciously controlled throughout their life cycle.
Facing this challenge, we propose RecAGT, a novel identification scheme aimed at reducing communication overhead and identifying potential malicious nodes. 
First, shard testable codes are designed to encode the original data in case of a leak of confidential data. Second, a new identity proof protocol is presented as evidence against malicious behavior. Finally, adaptive group testing is chosen to identify malicious nodes.
Notably, our work focuses on the internal operation within the committee and can thus be applied to any sharding permissioned blockchains. 
Simulation results show that our proposed scheme can effectively identify malicious nodes with low communication and computational costs.
\end{abstract}

\section{Introduction}
Permissioned blockchain has emerged as an appropriate architecture concept for business environments, and it is presently arising as a promising solution for distributed cross-enterprise applications.
However, it still faces many challenges regarding verifiability\cite{rebello2022security,amiri2021separ}, scalability\cite{amiri2021permissioned}, and performance\cite{gorenflo2020fastfabric}. To solve these problems, the sharding technique inspired by Spanner\cite{corbett2013spanner} is proposed to be integrated with permissioned blockchain, which partitions block data into multiple shards that are maintained by different committees (or ``clusters'').

Existing work\cite{amiri2019sharding,amiri2021sharper,dang2019towards,huang2022elastic,gao2022pshard,mao2023geochain,zheng2022meepo} in sharding permissioned blockchains focuses on how to partition nodes into different committees and efficiently handle cross-shard transactions. But they ignore \textbf{the vulnerability of nodes to malicious control}. 
There is no guarantee that nodes could remain honest\footnote{Honest nodes are those that perform normally following the rules of the system (\eg read, write or maintain blocks and perform or relay transactions).}.
In other words, nodes cannot always behave normally throughout their life cycle.
For example, nodes could come under control and turn malicious\footnote{The behaviors of malicious (or Byzantine) nodes could censor, reverse, reorder or withhold specific transactions without including them in any block to interfere with the system\cite{falazi2020transactional}.} as a result of cyber attacks such as BGP hijacking\cite{ekparinya2019attack}, DNS attack\cite{saad2020exploring}, or Eclipse attack\cite{davenport2018attack}.
Many of the previous research\cite{yu2019lagrange,solanki2019non,hong2022byzantine,hong2022hierarchical} on malicious node identification has been explored in distributed computing. The common idea of them is to utilize different coding algorithms to check the final computation output and use numerous testing trials to find malicious nodes.
However, workers in distributed systems perform intermediate computing tasks and do not maintain any data locally. Moreover, there must be complete trust between the master and workers\cite{zhao2020blockchain}.

In this paper, we consider the node vulnerability problem in sharding permissioned blockchains. We propose the sha\textbf{R}d t\textbf{e}stable \textbf{c}ode with \textbf{A}daptive \textbf{G}roup \textbf{T}esting (\textbf{RecAGT}), a malicious node identification scheme. Specifically, we first present our shard testable codes by designing polynomial functions to reduce communication overhead. Nodes perform verification based on the properties of testable codes. Then we design an identity proof protocol based on the digital signature as the proof of malicious behaviors. Finally, we use an adaptive group testing algorithm to calculate the required number of test trials. Therefore, the newly-joined node can verify the received messages and recover the original data to improve its ability to identify malicious nodes, which further enhances the security and stability of the sharding permissioned blockchain.

The main contributions of this paper are summarized as follows:
\begin{itemize}
  \item We propose a new scheme called RecAGT for malicious node identification. It is shown that communication costs and computational complexity will be significantly reduced from $O(n^{2}b)$ to $O(\log^2(m)\log\log(m))$ compared to other schemes (Table \ref{table:complexity}).
  \item In addition, the administrator could perform adaptive group testing to reduce the number of tests required to identify malicious nodes.
  \item We conduct theoretical simulations and the results show that our scheme only needs a low number of group tests, which effectively improves the system security and stability.  
\end{itemize}

\begin{table}[htbp]
    \renewcommand\arraystretch{1.1}  %
    \tabcolsep=1em   %
    \centering
    \begin{threeparttable}
    \caption{The communication cost and computational complexity of nodes\tnote{1}}
    \begin{tabular}{ccc}
        \toprule
        \textbf{} & \textbf{Communication cost} & \textbf{Computational complexity} \\
        \midrule
        \textbf{Uncoded}  & $O(nb)$   & $O(n^{2}b)$   \\
        \textbf{CheckSum}  & $O(b)$   & $O(n^{2})$   \\
        \textbf{RecAGT}  & $O(b)$  &  $O(\log^2(m)\log\log(m))$  \\
        \bottomrule
    \end{tabular}
    \label{table:complexity}

    \begin{tablenotes}
    \footnotesize
    \item[1] $b$: shard size in bytes, $n$: size of committee, $m$: size of coding shard
    \end{tablenotes}
    \end{threeparttable}
\end{table}

The rest of the paper is organized as follows. Section \ref{section:RelatedWork} discusses related work. Section \ref{section:SystemModel} describes the setup of the permissioned blockchain and introduces the system model. In Section \ref{section:method}, we propose our identification scheme and give detailed theoretical analyses. Section \ref{section:experiments} analyzes and discusses the experimental results. Finally, Section \ref{section:conclusion} concludes the paper.
\section{Related Work}
\label{section:RelatedWork}
Recently, a lot of research has been made on the permissioned blockchain to improve its verifiability, scalability, and performance. Since our scheme is proposed based on the sharding technique, we will discuss the related work in the aspects of sharding and other methods.

\textbf{Sharding methods.}
There have been many studies working on sharding permissioned blockchains. Amiri \etal \cite{amiri2019sharding} introduce a model to handle intra-shard transactions and their subsequent work\cite{amiri2021sharper} uses a directed acyclic graph to resolve cross-shard transaction agreements to improve verifiability. Dang \etal\cite{dang2019towards} design a comprehensive protocol including shard formation and transaction handling to upgrade performance.
Huang \etal \cite{huang2022elastic} propose an adaptive resource allocation algorithm to efficiently allocate network resources for system stability. 
Gao \etal \cite{gao2022pshard} propose the Pshard protocol, which adopts a two-layer data model and uses a two-phase method to execute cross-shard transactions to ensure safety and liveness. 
Mao \etal \cite{mao2023geochain} propose a locality-based sharding protocol in which they cluster participants based on their geographical properties to optimize inter-shard performance. 
As we can see, they pay more attention to the operation of blockchain systems. Nonetheless, these approaches overlook the potential actions of individual nodes. In our research, we consider the scenario where nodes might be under malicious control and propose an identification scheme to mitigate these vulnerabilities.

\textbf{Identification of malicious nodes.}
The problem of malicious node identification has been studied in distributed systems. 
Yu \etal \cite{yu2019lagrange} provide resiliency against stragglers and security against Byzantine attacks based on Lagrange codes. 
Solanki \etal \cite{solanki2019non} design a coding scheme to identify attackers in distributed computing.
Hong \etal \cite{hong2022byzantine} propose locally testable codes to identify Byzantine attackers in distributed matrix multiplication. 
They also propose a hierarchical group testing\cite{hong2022hierarchical} in distributed matrix multiplication, making the required number of tests smaller. However, there must be complete trust between the master and workers, which is unsuitable for blockchain. In this paper, we develop an identity proof method that serves as a safeguard against potential malicious behavior.

In view of these unresolved problems, we propose a novel identification scheme named RecAGT, specifically designed to address the issue of identifying malicious nodes effectively.

\begin{remark}
Our work focuses on reducing communication costs and identifying potential malicious nodes based on their transmitted messages. Therefore, we do not investigate further details about transaction verification and subsequent penalty actions.
\end{remark}
\section{System Overview} 
\label{section:SystemModel}
In this section, we introduce the system model based on the permissioned\footnote{Generally speaking, permissioned blockchains can be divided into private and consortium blockchains since both of them only allow nodes with identity to join the network. Our study primarily focuses on consortium blockchains due to their alignment with the idea of decentralization.} blockchain and explore potential corresponding attacks. 

\begin{figure}[H]
    \centering
    \includegraphics[width=1.0\textwidth]{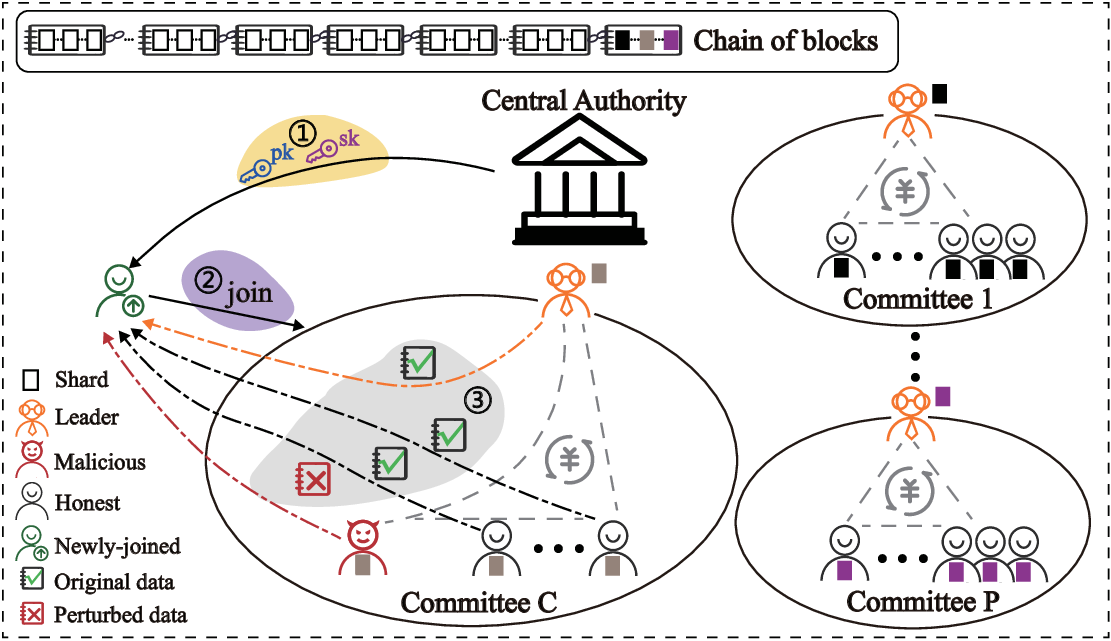}
    \vspace{-2.0em}
    \caption{System model of sharding permissioned blockchain under the existence of malicious nodes}
    \label{fig:system_model}
\end{figure}

\subsection{System Model}
The permissioned blockchain is a distributed ledger that cannot be publicly accessed and is only open to users with authorized digital certificates. Nodes perform specific operations granted by the administrator (\textit{Central Authority, CA}). Without losing generality, we assume that there is a \textit{public key infrastructure} (\textit{PKI}) in the system, that is, \textit{CA} distributes the private (secret) key $sk_{i}$ and the public key $pk_{i}$ to each node $\rm{N}$$_{i}$ as identity credentials. Note that each node knows each other's public key via \textit{CA}. In addition, \textit{CA} needs to assign a unique scalar $x_{i}$ to each member for the construction of our identification scheme.

In sharding permissioned blockchains, nodes are partitioned into committees. Essentially, the processing mechanism is the same for each committee since each can be seen as a tiny blockchain system that maintains a subchain. Hence, we will focus on a committee $\mathcal{C}$ in the following sections.
There are basically two types of nodes: full nodes and light nodes. Full nodes are able to generate and store valid blocks and verify new blocks from other full nodes. Light nodes are able to perform transaction inclusion verification and thus increase blockchain scalability instead of storing the entire ledger. Based on the purpose of reducing communication costs and identifying malicious nodes, the object of our study is the full node.

The overall system model is illustrated in Fig.~\ref{fig:system_model}. Each member would maintain the entire shards of its corresponding committee under a permissioned blockchain. But for simplicity of presentation, we only show the shard of the latest block for each node. Each committee handles different transactions in parallel. When a new node enters the network, \textit{CA} will verify its identity and assign keys ($sk$ and $pk$) to it (shown in yellow background). Then the new node joins one of the committees by partitioning rules (shown in purple background), and it needs to retrieve all the shards of the committee in order to participate in the transaction processing. However, the node cannot fully trust any members, even the leader, at all. Therefore, it must receive data messages from as many members as possible against interference from potential malicious nodes (shown in grey background).
 
The data stored by nodes is a chain of sub-blocks, denoted as $B$, which is actually a set of byte strings of blocks containing a batch of transactions. If we divide them into $m$ subshards, they can be shown as
$$
B = \begin{bmatrix}
B_{1} & B_{2} & \cdots & B_{m}
\end{bmatrix}^{\top}.
$$

The system has $P$ committees and each committee $\mathcal{C}$ includes $n$ members.
Other assumptions of our network model are as follows:
\begin{enumerate}
  \item Credibility of nodes: we assume only \textit{CA} is honest and has no assumptions about other nodes. It's a weak decentralization. In practice, permissioned blockchains do not defy the principles of decentralization but rather strike a balance between centralized and decentralized requirements.
  \item Finite maximum network delay: all requests can be answered in a finite time. In other words, a finite maximum delay $\delta$ is assumed. If a node sends a request, it will receive a response message within $\delta$. Otherwise, we assume that the target node is offline or does not work.
  \item Network communication:
  \begin{itemize}
    \item Nodes communicate with other nodes;
    \item Point-to-point communications are asynchronous;
    \item The communication channel is noiseless.
  \end{itemize}
\end{enumerate}

\subsection{Threat Model}
When a new node enters the system, its identity undergoes verification by \textit{CA}. Once the node's identity is authenticated, \textit{CA} allocates identity credentials to the new node. However, it is possible that committee members behave maliciously to prevent new nodes from joining the system (\eg they may not send feedback or they may respond to perturbed shard data). In other words, although nodes are authenticated when they join the network, they cannot be fully trusted because there is no guarantee that they will not be maliciously controlled throughout their life cycle. Therefore, it is imperative to devise an efficient identification scheme to identify potential malicious nodes.
\begin{figure}[htb]
    \centering
    \includegraphics[width=\textwidth]{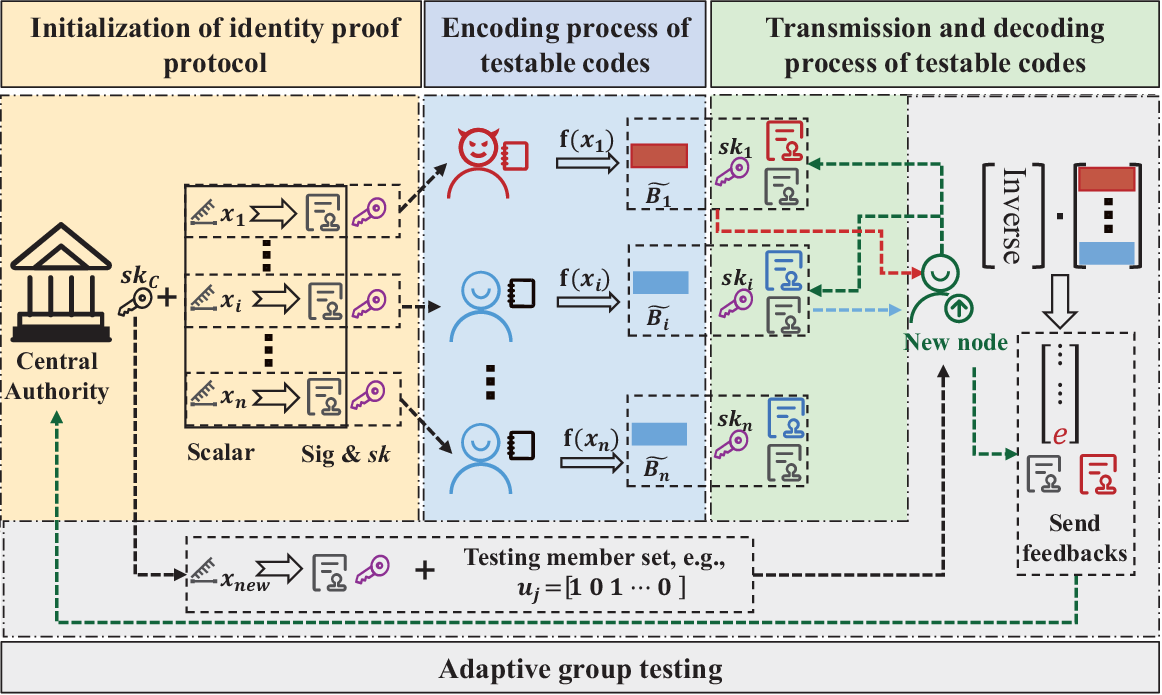}
    \vspace{-2.0em}
    \caption{Overview of the RecAGT scheme. When a new node $\rm{N}$$_{new}$ joins the system, \textit{CA} utilizes \textit{PKI} to allocate public and private keys, and a scalar along with a testing set to it. \textit{CA} signs the scalar of $\rm{N}$$_{new}$ using its own private key and transmits it to $\rm{N}$$_{new}$. This allows $\rm{N}$$_{new}$ to verify and store the received scalar using the \textit{CA}'s public key. Then, $\rm{N}$$_{new}$ sends requests to nodes of the testing set. $\rm{N}$$_{new}$ generates the historical data of its committee by receiving messages containing encoded shard data with the first and second signatures. If the computing result is incorrect during the verification process, $\rm{N}$$_{new}$ will forward feedback to \textit{CA} based on these messages, aiding in identifying potential malicious nodes.}
    \label{fig:plan_overview}
\end{figure}

\section{RecAGT Identification Scheme}
\label{section:method}
Building upon the system model outlined in Section \ref{section:SystemModel}, we propose an identification scheme based on our shard testable codes. In this scheme, newly-joined nodes are provided with encoded shard data to facilitate the retrieval of the original data specific to the committee they are part of. If an inconsistency arises during the verification process, the node will send feedback to \textit{CA} to help identify potential malicious nodes, as illustrated in Fig.~\ref{fig:plan_overview}. The scheme consists of three key components: shard testable codes, an identity proof protocol, and an adaptive group testing method.

A straightforward data identification scheme, referred to as ``\textbf{Uncoded}'', might involve each committee member sending the original data to the newly-joined node, enabling its participation in blockchain activities.
However, malicious nodes could disrupt the normal operation (\eg sending fraudulent transactions). Therefore, the node has to compare all received data to ensure that there is no perturbed data, which would be equivalent to receiving the entire shard. If MD5 is adopted to perform the data consistency check, then the computational complexity for newly-joined nodes will be $O(n^{2}b)$ which is time-consuming and increases the communication overhead.

\textbf{Straw man method: ``\textbf{Checksums}''.}
There is no need to ask each member to send original data. Intuitively, we could replace those original data messages with checksums. Checksums (such as SHA-256) are used to check the integrity of web content so that it can determine if a document has changed in case of a tampering attack. Specifically, when a new node joins the network, only one of the members needs to send the original data. What the remaining members need to send is the digest generated by the checksums method from the shard they store locally. Since the newly-joined node does not need to generate the digest itself, the computational complexity is $O(n^{2})$. 

The above methods are not efficient on account that they require pairwise comparisons for each message. Moreover, they are vulnerable when it comes to privacy-preserving since both send at least one original data to an unidentified node. However, in a permissioned blockchain, the internal data of each organization should be kept confidential meanwhile cross-enterprise transactions should be transparent to all parties. Therefore, in addition to identifying potential malicious nodes, we need to encode data during transmission.

\subsection{Shard Testable Codes}
To check if the returned data is not perturbed in an efficient way, we first design a method to encode the data that is inspired by polynomial codes\cite{yu2017polynomial}. The linear encoding function $\mathbf{f}$ is constructed by using the row-divided sub-matrices $B_{v}$ and the corresponding scalars $x_{i}$ as coefficients, which is given by \begin{equation}
\label{f_A}  %
\mathbf{f}(x_{i})=\sum_{v=0}^{m-1} B_{v+1}^T x_i^v = \widetilde{B}_{i}.
\end{equation}
Accordingly, each member will store the encoded shard (\eg Node ${\rm{N}}_i$ will store $\widetilde{B}_{i}$ locally).

By combining encoding functions for a committee of size $n$, we represent the coded shards as a matrix, which is given by
\begin{equation}
\label{encoding_matrix}  %
\widetilde{B}= G\cdot B = \begin{bmatrix}
x_{1}^{0} & \cdots & x_{1}^{m-1} \\
\vdots & \ddots  & \vdots \\
x_{n}^{0} & \cdots & x_{n}^{m-1} \\
\end{bmatrix} \cdot \begin{bmatrix}
B_{1} \\ \vdots \\ B_{m}
\end{bmatrix} = \begin{bmatrix}
\widetilde{B}_{1} \\ \vdots\\ \widetilde{B}_{n}
\end{bmatrix},
\end{equation}
where $G$ is the encoding matrix consisting of all scalars assigned to each committee member.

Waiting for $m+1$ encoded shard messages from different nodes to form a test group denoted $\mathbf{u}$ (we will discuss the test group in Section \ref{section:adaptive_group_testing}), the newly-joined node collects these messages into a new matrix $C$. For simplicity of demonstration, we assume that the nodes' indexes that return messages are [$1$: $m+1$], which means that their scalars are $[x_{i}]_{i=1}^{m+1}$. Therefore, the result vector consisting of intermediate messages can be expressed as
\begin{equation}
 \begin{aligned}
 \centering
 \label{eqn_c}
 C&=G_{u} \cdot B 
 =\left[\begin{array}{ccc}
  x_1^0 & \cdots & x_1^{m-1} \\
  \vdots & \ddots & \vdots \\
  x_{m+1}^0 & \cdots & x_{m+1}^{m-1}
 \end{array}\right] \cdot B,\\
\end{aligned}
\end{equation}
where
\begin{equation}
 \begin{aligned}
 \centering
  C_{i}&= \sum_{k=0}^{m-1}{x_{i}^{k}B_{k+1}},\, \forall i \in [1,\,m+1],
  \end{aligned}
\end{equation}
and $G_{u}$ is the encoding matrix consisting of the $m+1$ encoding vectors.

Based on $G_{u}$, we can make Vandermonde matrix $V$ by adding one more column shown in bold, which is given by
\begin{equation}
\label{V}
V= \begin{bmatrix}
x_{1}^{0} & \cdots & x_{1}^{m-1} & \boldsymbol{x_{1}^{m}} \\
\vdots & \ddots  & \vdots & \boldsymbol{\vdots} \\
x_{m+1}^{0} & \cdots & x_{m+1}^{m-1} & \boldsymbol{x_{m+1}^{m}}\\
\end{bmatrix}.
\end{equation}

Since the Vandermonde matrix is invertible, we denote the inverse of $V$ as $S$. The inverse of the Vandermonde matrix can be computed as
\begin{equation}
    S = \prod_{1 \leq j<i \leq m+1} \frac{1}{x_i-x_j}.
\end{equation}
Since the computing process can be viewed as a polynomial interpolation problem\cite{verde1988inverses}. We adopt one of the  method\cite{kedlaya2011fast} whose computational complexity is $O(\log^2(m)\log\log(m))$. Note that the complexity could be further reduced by adopting any variant of interpolation algorithms.

Without loss of generality, let's define matrix $U$ as the first $m$ columns of $V$, which is another name for $G_{u}$. Therefore, $C=U\cdot B=G_{u}\cdot B$, and by the construction and rules for matrix multiplication, the following equation holds true,
\begin{equation}
\label{eq_SU}
S\cdot U= \left[\frac{\mathbf{I}_{m}}{\mathbf{0}_{1, m}}\right],
\end{equation}
where $\mathbf{I}_{m}$ denotes the identity matrix of rank $m$ and $\mathbf{0}_{1, m}$ denotes the null matrix of row size $1$ and column size $m$. Let's denote $S_{j,\,k}$ as the element of row $j$ and column $k$ in matrix $S$, we can express \eqnref{eq_SU} as 
\begin{equation}
\label{eq_SU_detail}
\sum_{j=1}^{m+1}S_{m+1,\,j}\times x_{j}^{p} = 0, \; \forall p \in [0, m-1].
\end{equation}

After the construction of shard testable codes, we will discuss two different scenarios based on the intermediate messages: 
\begin{enumerate}
    \item These messages are all correct, which means the test group set is honest; 
    \item There is at least one perturbed message, which means some of the group members are malicious.
\end{enumerate}

\textbf{The group members are all honest.}
If the test group does not include malicious nodes, with the group and the collected messages serving as $\mathbf{u}_{h}$ and $C_{h}$, respectively, we would get the following test result
\begin{equation}
\label{SC}
S \cdot C_{h}=S\cdot U\cdot B = 
\left[\frac{\mathbf{I}_{m}}{\mathbf{0}_{1,\,m}}\right] \cdot B,
\end{equation}
which means $S_{m+1} \cdot C=\mathbf{0}_{1,\,m}$ where $S_{m+1}$ is the row $m+1$ of matrix $S$. By using \eqnref{eq_SU} and (\ref{eq_SU_detail}), we can define the output of test row result as $\mathbf{O}=S_{m+1} \cdot C$.

We choose any $m$ rows of matrix $C$, denoted as $\overset{m}{C}$, and the corresponding encoding matrix is denoted as $\overset{m}{G}$. Therefore, we can get 
\begin{equation}
\label{eqn_C^m}
\overset{m}{C} = \overset{m}{G} \cdot B.
\end{equation}
Since $\overset{m}{G}$ is a Vandermonde matrix, we represent the inverse of $\overset{m}{G}$ as $(\overset{m}{G})^{-1}$. When we left multiply both sides of \eqnref{eqn_C^m} by the inverse of $\overset{m}{G}$, we would get
\begin{equation}
\label{eqn_recover_block}
(\overset{m}{G})^{-1} \cdot \overset{m}{C} = (\overset{m}{G})^{-1} \cdot \overset{m}{G} \cdot B = B,
\end{equation}
which means the newly-joined node could recover the original shard based on our coding scheme.

In the following scenario, we will consider a more complicated case where some malicious nodes will interfere with the shard data to keep the new node from joining the committee, which compromises the scalability of the permissioned blockchain.

\textbf{The group includes at least one malicious node.}
If there is at least one byzantine node in the group, denoting the group as $\mathbf{u}_{b}$ and the matrix including perturbed shard data as $C_{b}$, we would get the following result
\begin{equation}
\label{eq_SC_byzantine}
S \cdot C_{b}=
\left[\frac{\mathbf{I}_{m}}{\mathbf{e}_{1, m}}\right] \cdot B,
\end{equation}
where $\mathbf{e}_{1, m} \neq \mathbf{0}_{1, m}$, which indicates some of the elements in matrix $\mathbf{e}_{1, m}$ are not zero. We can express equation (\ref{eq_SC_byzantine}) as
\begin{equation}
\label{eq_SC_byzantine_detail}
\sum_{j=1}^{m+1}S_{m+1, j}\times x_{j}^{p} = 0 + b, \quad \forall p \in [0, m-1],
\end{equation}
where $b$ is the interfering data.
\begin{lemma}
\label{lemma_O}
The output of the test result in the last row with malicious group $\mathbf{u}_{b}$ is $\mathbf{O}\neq0$, and the output with honest group $\mathbf{u}_{h}$ is $\mathbf{O} = 0$.
\end{lemma}
\begin{proof}
    Based on our construction, the core of our test computation is the last element, denoted as $\mathbf{O}=S_{m+1} \cdot C$. If the group is a malicious group $\mathbf{u}_{b}$, then the result is
    \begin{eqnarray}
        \mathbf{O} &=&S_{m+1} \cdot C_{b} \nonumber \\
        &=& \sum_{j=1}^{m+1}S_{m+1,\,j}\cdot C_{j} \nonumber  \\
        &=& \sum_{j=1}^{m+1}S_{m+1,\,j} \left(\sum_{k=0}^{m-1}x_{j}^{k}B_{k+1} + b_{j}\right)  \nonumber \\
        &=& \sum_{j=1}^{m+1}\left(\sum_{k=0}^{m-1}{S_{m+1,\,j}x_{j}^{k}B_{k+1}} \right) + \sum_{j=1}^{m+1}S_{m+1,\,j} b_{j} \label{proof_O_term_e} \\
        &=&  \sum_{j=1}^{m+1}S_{m+1,\,j} b_{j}. \quad (\forall k \in [0, m-1]).
    \end{eqnarray}
    However, it is possible that the term $\sum_{j=1}^{m+1}S_{m+1,\,j} b_{j}$ in (\ref{proof_O_term_e}) could be zero even if the group has malicious nodes. Under this situation and over a finite field $\mathbb{F}_{q}$, the probability of this exceptional case to occur is at most $1/q$. Thus, the probability approaches zero as the field size $q$ increases.
    
    By contrast, if the group members are all honest, the term $\sum_{j=1}^{m+1}S_{m+1,\,j} b_{j}$ in (\ref{proof_O_term_e}) will not exist. Therefore, the result with a honest group $\mathbf{u}_{h}$ will be $\mathbf{O}=S_{m+1} \cdot C_{h}=0$.
\end{proof}

\subsection{Identity Proof Protocol}
We develop an identity proof protocol using the digital signature to ensure the integrity of the node's scalar. What's more, it can be used as evidence if a node perturbs or forges data, in which case the new node cannot compute the inverse of the Vandermonde matrix causing the failure to recover the original shard data.

Digital signature technology is used to encrypt the digest of a message with the sender's private key and transmit it to the receiver along with the plain message. The receiver can only decrypt the encrypted digest with the sender's public key to get $d_{new}$, and then use the hash function to generate a digest $d_{ori}$ for the received plain text. If $d_{new}$ and $d_{ori}$ are the same, it means that the plain message received is complete and not modified during the transmission, otherwise the message is tampered with. Thus the digital signature can be used to verify the integrity of the information.

Inspired by \cite{kaur2012digital} and based on the characteristic of known public keys in permissioned blockchain, our identity proof protocol using digital signature contains the following components:
\begin{itemize}
    \item The message $\mathcal{M}$, which is the content to which the signature algorithm may be applied. The content in our protocol is "node's scalar $x_{i}$ where $i$ is the index of the node;
    \item A key generation algorithm $G$, which is used by the central authority to generate credentials for each node;
    \item A signature algorithm $\sigma$, which produces a signature $\sigma (\mathcal{M}, sk_{i})$ for a message $\mathcal{M}$ using the secret key $sk_{i}$;
    \item A verification algorithm $\mathcal{V}$, which tests whether $\sigma (\mathcal{M}, sk_{i})$ is a valid signature for message $\mathcal{M}$ using the corresponding public key $pk_{i}$. In other words, $\mathcal{V}(\sigma, \mathcal{M}, pk_{i})$ will be $\mathrm{true}$ if and only if it is valid.
\end{itemize}

Since \textit{PKI} exists among nodes, each node $i$ could create a digital signature $\sigma (\mathcal{M}, sk_{i})$ on message $\mathcal{M}$ with its secret key $sk_{i}$. And the signature can be verified by the corresponding public key $pk_{i}$ which is known to each node in the committee.

\textbf{Cryptographic Primitives} First, we present some primitives that we use in the rest of the paper.
\begin{itemize}
    \item \hash{msg} - a cryptographically secure hash function that returns the digest of $msg$ (e.g., SHA-256, SHA-512);
    \item \encrypt{hash, sk} - an encrypted hash function that returns the encrypted hash value (or called signature) for a hash value $hash$ using a secret key $sk$ ;
    \item \decrypt{sig, pk} - a decrypted hash function that returns the hash value of signature result using corresponding public key $pk$.
\end{itemize}

\textbf{Signature verification}
At the initialization phase of each committee, \textit{CA} will send the plain message $\mathcal{M}$ (which is scalar $x_{i}$) and signature $sig_{C}^{i} = \sigma (\mathcal{M}, sk_{C})$ to each committee member, where $i$ is the index of the target node and $sk_{C}$ is the secret key of \textit{CA}. Since the setting of permissioned blockchain where member knows the public key of each other, the member could use \textit{CA}'s public key $pk_{C}$ to verify if $sig_{C}^{i}$ is valid using $\mathcal{V}(sig_{C}^{i}, \mathcal{M}, pk_{C})$.
More details are shown in Algorithm \ref{algo_veri_scalar}. The necessity for a two-step comparison can be attributed to two distinct purposes. First, the first signature guarantees that scalars from other nodes are allocated by \textit{CA}. Next, the second signature is critical in confirming that the first signature is sent from the intermediary node. An erroneous second signature allows $\rm{N}$$_{new}$ to forward the message to \textit{CA}, offering ``evidence'' to reveal any suspicious behavior, as the private key necessary for signing the second message is unique to the intermediary node.

\subsection{Adaptive Group Testing Method}

\label{section:adaptive_group_testing}
Group testing originated from World War II for testing blood supplies, Robert Dorfman reduced the number of tests detecting whether the US military draftees had syphilis dramatically by pooling samples\cite{dorfman1943detection}.

There are two types of group testing methods for identifying members with defects in a group: non-adaptive group testing (NAGT) and adaptive group testing (AGT). NAGT involves pooling samples from multiple individuals and testing them all together as a group, while AGT involves designing the test pools sequentially and adjusting the groups based on previous test results to minimize the number of individual tests needed. 
We adopt the AGT method to identify potential malicious nodes since NAGT requires a large number of tests, which is time-consuming, and the validation of transactions in the blockchain is very sensitive to time. 

\begin{algorithm}[!htb]
    \caption{Identity proof protocol with digital signature} %
    \label{algo_veri_scalar} %
    \SetKwInput{PhaseOne}{$\triangleright$ Phase 1}
    \SetKwInput{PhaseTwo}{$\triangleright$ Phase 2}
    \SetKwInput{PhaseThree}{$\triangleright$ Phase 3}
    \SetKwBlock{AsLeader}{As a leader $\rm{N}$$_{leader}$}{}
    \SetKwBlock{AsNode}{As a node $\rm{N}$$_{i}$}{}
    \SetKwBlock{AsOtherNode}{As other node $\rm{N}$$_{i}$ of the committee}{}
    \SetKwBlock{AsCA}{As $\textbf{\textit{CA}}$ of the system}{}
    \PhaseOne{Initialization of nodes}
    \AsCA{
        \ForEach{\rm{node} $\rm{N}$$_{i}$ \rm{in committee} $\mathcal{C}$}{
            $x_{i} \leftarrow$ generate a random scalar over $\mathbb{F}_{q}$\;
            $sig_{C}^{i} \leftarrow$ $\sigma (x_{i}, sk_{C})$\tcp*{encrypt the scalar $x_{i}$ with $sk_{C}$} 
            send [$x_{i}$, $sig_{C}^{i}$] to node $\rm{N}_{i}$\;
        } 
    }
    \AsOtherNode{
        wait for message [$x_{i}$, $sig_{C}^{i}$] from $\textbf{\textit{CA}}$\;
        $pk_{C} \leftarrow$ query public key of $\textbf{\textit{CA}}$\;
        \eIf{$\mathcal{V}(sig_{C}^{i}$, $x_{i}$, $pk_{C})$ == $\rm{true}$}{
            store the assigned scalar $x_{i}$\;
        }(False){
            request $\textbf{\textit{CA}}$ to resend the message\; %
        }
    }
    \PhaseTwo{New node $\rm{N}$$_{new}$ joins the committee $\mathcal{C}$}
    \AsCA{
        $x_{new} \leftarrow$ generate a new random scalar over $\mathbb{F}_{q}$\;
            $sig_{C}^{new} \leftarrow$ $\sigma (x_{new}, sk_{C})$\; 
            send [$x_{new}$, $sig_{C}^{new}$] to node $\rm{N}$$_{new}$\;
    }
    \AsOtherNode{
        $sig_{i}^{new}\leftarrow$ $\sigma (sig_{C}^{i}, sk_{i})$\tcp*{new signature for $sig_{C}^{i}$ with $sk_{i}$}
        send [$x_{i}$, $sig_{C}^{i}$, $sig_{i}^{new}$] to node $\rm{N}_{new}$\;
    }    
    \PhaseThree{Signature verification by newly-joined node $\rm{N}$$_{new}$}
    wait for messages from nodes of the committee $\mathcal{C}$\;
    \ForEach{\rm{message} [$x_{i}$, $sig_{C}^{i}$, $sig_{i}^{new}$] from node $\rm{N}_{i}$ }{
        $pk_{i} \leftarrow$ query public key of $\rm{N}$$_{i}$ from $\textbf{\textit{CA}}$\;
        $pk_{C} \leftarrow$ query public key of $\rm{N}$$_{leader}$ from $\textbf{\textit{CA}}$\;
        \uIf{
        \Fverify{$sig_{i}^{new}$, $sig_{C}^{i}$, $pk_{i}$} $\wedge$ \Fverify{$sig_{C}^{i}$, $x_{i}$, $pk_{C}$} }{
            store the scalar $x_{i}$ locally for decoding operation\;
            }\uElseIf(\tcp*[f]{the second signature is invalid, there may be some error during transmission}){$\neg$(\Fverify{$sig_{i}^{new}$, $sig_{C}^{i}$, $pk_{i}$})}{
            require $\rm{N}$$_{i}$ to resend the message\;
            }\Else(\tcp*[h]{the initial signature is inconsistent with $x_{i}$}){
            send [$x_{i}$, $sig_{C}^{i}$, $sig_{i}^{new}$] to \textbf{\textit{CA}}\tcp*{make it as fraud-proof to punish the sender $\rm{N}$$_{i}$}
            }
    }
\end{algorithm}

In our design, we assume each node is unidentified unless there is an identity proof to ensure its credibility, which means these nodes will be in at least one negative test. The goal is to build a testing set with as few tests $T$ as possible to identify all malicious nodes. In particular, each test $\mathbf{u}$ consists of a group of nodes, and the result is false if data messages transmitted are tampered with. Otherwise, it is correct.
Based on our settings, the group testing problem can be formulated as follows
\begin{equation}
\label{GT_formalization}
\mathbf{y}=\mathbf{M} \circ \mathbf{x},
\end{equation}
where (1) $\circ$ denotes the row-wise Boolean operation, and the result is $y_{i} = 0$ if nodes in the $i$-th test are all honest or $y_{i} \neq 0$ instead;
(2) $\mathbf{M} \in \mathbb{F}_{2}^{t\times n}$ is a contact matrix where $\mathbf{M}_{i, j} = 1$ indicates the $i$-th test contains node $\mathrm{N}$$_{j}$.

The goal of group testing is to design a test matrix $\mathbf{M}$ such that the number of tests is as small as possible and it can be expressed as $\mathbf{M}=[\mathbf{u_{1}}\;\mathbf{u_{2}}\;\cdots \; \mathbf{u_{n}}]^{\top}$. 
Items included in tests with negative outcomes will be viewed as noninfective and collected into an honest set $\mathcal{H}$. Similarly, those items in positive tests (validation result is wrong) will be collected into the malicious set $\mathcal{S}$. Owing to the property of $k$-disjunct, each sample includes a different set of tests. By matching the honest set $\mathcal{H}$ and malicious set $\mathcal{S}$, the $f$ defectives can be identified.

The selection of nodes in a test to identify malicious nodes with shard testable codes can be regarded as a group testing problem. Thus the $\circ$ operation in \eqnref{GT_formalization} can be expressed as the matrix multiplication operation of shard messages verification in \eqnref{eq_SU}. And our goal is to minimize the number of tests for identifying malicious nodes. In the general AGT problem, the algorithm designs a set of tests \{$u_{1}$, $u_{2}$, $\cdots$, $u_{T}$\} to make $T$ as small as possible. Given the outcomes of these tests, the honest set and malicious set will be generated, achieving the goal of malicious node identification.

In a simple method, \textit{CA} can fix $m$ honest nodes and check one unidentified node whether honest or not if \textit{CA} knows $m+1$ honest nodes with trials before. Thus the remaining $N-m-1$ nodes can be identified in the same way. Dorfman\cite{dorfman1943detection} proposes a simple procedure which partitions the whole $E$ items containing $f$ defectives into $\sqrt{E f}$ subsets, each of size $\sqrt{E/f}$. Hence, the number of tests that Dorfman's procedure requires is at most
\begin{equation}
    T=\sqrt{Ef} + f\sqrt{\frac{E}{f}} = 2\sqrt{Ef}.
\end{equation}
Following Dorfman's procedure, the key is the number of trials $\widehat{T}$ that finds the first honest group whose result of the last row in \eqnref{eq_SU} is $\mathbf{O} = 0$. 

\begin{theorem}
\label{theorem:prob_no_malicious}
Given a committee of $n$ nodes, if there are $f$ malicious nodes where $n \geq m + f + 1$, the probability of having no malicious nodes in a test group of size $m+1$ is
\begin{equation}
\label{prob_no_malicious}  %
P(H=0) = \prod\limits_{i=0}\limits^{m} (1-\frac{f}{n-i}),
\end{equation}
where $H$ is the number of malicious nodes.
\end{theorem}
\begin{proof}
    We first compute the number of ways to pick $m+1$ chunks among the set of non-malicious nodes, \ie $\binom{n-f}{m+1}$. Similarly, we can produce the total number of ways to pick any $m+1$ samples out of the total number of samples, \ie $\binom{n}{m+1}$. Therefore, the probability can be computed as
    \begin{eqnarray}
    \label{eqn:P_H_0}
      P(H=0) &=& \frac{\binom{n-f}{m+1}}{\binom{n}{m+1}}  \nonumber   \\
      &=& \prod\limits_{i=0}\limits^{m} (1-\frac{f}{n-i}). \label{eqn_simple_form}
    \end{eqnarray}
\end{proof}

If we define \textit{Malice Ratio} ($R_{f}$) to represent the proportion of malicious nodes to the total number of nodes ($f/n$), then we can rewrite Equation \eqref{eqn_simple_form} as:
\begin{equation}\label{eqn_ratio_form}
P(H=0) = \prod_{i=0}^{m} \left(1 - \frac{R_{f}\times n}{n-i}\right).
\end{equation}
After thorough analysis and computation, we have observed that regardless of the total number of nodes within a committee, when the \textit{Malice Ratio} does not exceed $1/5$ and the number of encoding shards is limited to $2$ or fewer, there is at least a $50\%$ probability that all members of the testing set are honest nodes. Under such circumstances, system stability and security can be guaranteed.

Accordingly, the probability of failing to find the first honest group with $\widehat{T}$ trials is $(1-P(H=0))^{\widehat{T}}$. Assume there is a error probability $\rho$, and we want to make $(1-P(H=0))^{\widehat{T}} \leq \rho$. Since $0<(1-P(H=0))\leq 1$, $0<T\leq \log_{1-P(H=0)}{\rho}$. Therefore, the maximum number of trials to find the first non-malicious group with $\rho$ is 
\begin{equation}
    \widehat{T}=\log_{1-P(H=0)}{\rho}.
\end{equation}

\begin{theorem}
For malicious node identification in committees of sharding permissioned blockchain, adaptive group testing with shard testable codes can identify all malicious nodes by $T$ testing trials, where $T$ can be written as
    \begin{equation}
        T \leq \log_{1-P(H=0)}{\rho} + 2\sqrt{(n-m-1)f},
    \end{equation}
    where $P(H=0)$ is given in \eqnref{prob_no_malicious} of Lemma \ref{theorem:prob_no_malicious}.
\end{theorem}
\begin{proof}
    If \textit{CA} finds the first honest group of size $m+1$ in a committee by $\widehat{T}$ trials. Based on Dorfman's procedure, CA divides the remaining $n-m-1$ unidentified nodes into $\sqrt{(n-m-1)f}$ subgroups, each of size $\sqrt{(n-m-1)/f}$. For a newly-joined node, CA requires one of the subgroups to send shard messages to it so that the new node can perform verification and send feedback to \textit{CA}. With the results of group testing, \textit{CA} could test those suspicious nodes separately and identify them as malicious if the outcome is wrong. 

    Hence, by using Dorfman's procedure, the total number of group testing trials is at most 
    \begin{eqnarray}
        T &\leq& \widehat{T} + \sqrt{(n-m-1)f} + f\sqrt{\frac{n-m-1}{f}} \nonumber \\
        &=& \log_{1-P(H=0)}{\rho} + 2\sqrt{(n-m-1)f}. \label{dorfman_term}
    \end{eqnarray}
\end{proof}

\begin{remark}
The second term in \eqnref{dorfman_term} is from the original adaptive testing method\cite{dorfman1943detection}. It can be improved by other well-design schemes, \eg HGBSA\cite{hwang1972method}. But the core idea is similar. For simplicity of demonstration, we will not discuss the variant of adaptive group testing algorithms in this paper.
\end{remark}

\subsection{Cost and Complexity}
\textbf{Communication cost}
Following our RecAGT scheme, the newly-joined node has three parts of communication costs: (1) the initialization identity proof from \textit{CA}; (2) $m+1$ encoded shard messages from the assigned group testing members, and (3) $m+1$ identity proof messages consisting of the scalar $x_{i}$, the first signature from \textit{CA} and the second signature from the member. Thus the communication cost for a newly-joined node can be computed as
\begin{equation}
    \label{eqn:RecAGT_comu_cost}
    (w+z+s) + ((m+1)\times \frac{b}{m}) + ((m+1)\times (w+2z)) = O(b),
\end{equation}
where $w$ is the size of scalar, $z$ is the size of digital signature, $s$ is the size of secret key, $m$ is the size of testable codes and $b$ is the size of original shard.

\textbf{Computational complexity} At the core of computational complexity is the decoding complexity, \ie computing the inverse of the Vandermonde matrix. Therefore the computational complexity is $O(\log^2(m) \log\log(m))$.

\section{Experiments}
\label{section:experiments}
In this section, we conduct extensive experiments to evaluate the performance of our proposed identification scheme under different parameters. We also compare our identification scheme with others.

\subsection{Setup for Parameters}

\begin{table}[H]
    \vspace{-2.0em}
    \centering
    \caption{Different settings with respect to $P$, $n$, $f/n$ and m}
    \begin{tabular}{c|c|c|c|c}
    \hline
        \textbf{} & \tabincell{c}{\textbf{\# Committees} \\ \textbf{($P$)}} & \tabincell{c}{\textbf{\# each} \\ \textbf{committee($n$)}} & \tabincell{c}{\textbf{malicious ratio} \\ \textbf{($f/n$)}} & \tabincell{c}{\textbf{\# shard-coding} \\ \textbf{size ($m$)}}\\ \hline
        Setting 1 & 300 & 6 & 0.2 (1) & 2 \\ \hline
        Setting 2 & 70 & 24 & 0.125 (3) & 3 \\ \hline
        Setting 3 & 25 & 72 & 0.05 (4) & 8 \\ \hline
        Setting 4 & 4 & 450 & 0.01 (5) & 10 \\ \hline
    \end{tabular}
    \label{table:configuration}
    \vspace{-2.0em}
\end{table}
To simulate a more practical and realistic permissioned blockchain, we adopt a similar experimental configuration referring to PShard\cite{gao2022pshard} and Omniledger\cite{kokoris2018omniledger}. There are four settings under a network of $1800$ nodes, as shown in Table~\ref{table:configuration}. Specifically, the configuration of setting 3 means that there are $25$ committees of each $72$ members, and the adversary ratio is $0.05$ ($0.05\times 72\approx4$).

\subsection{Simulation Results and Analyses}

\begin{figure}[htbp]
    \centering  %
    \vspace{-0.35cm} %
    \subfigtopskip=2pt %
    \subfigbottomskip=2pt %
    \subfigcapskip=-5pt %
    \subfigure[Probability $P(H=0)$ with changing $m$]{   %
        \begin{minipage}[t]{0.45\linewidth}
        \includegraphics[width=\linewidth]{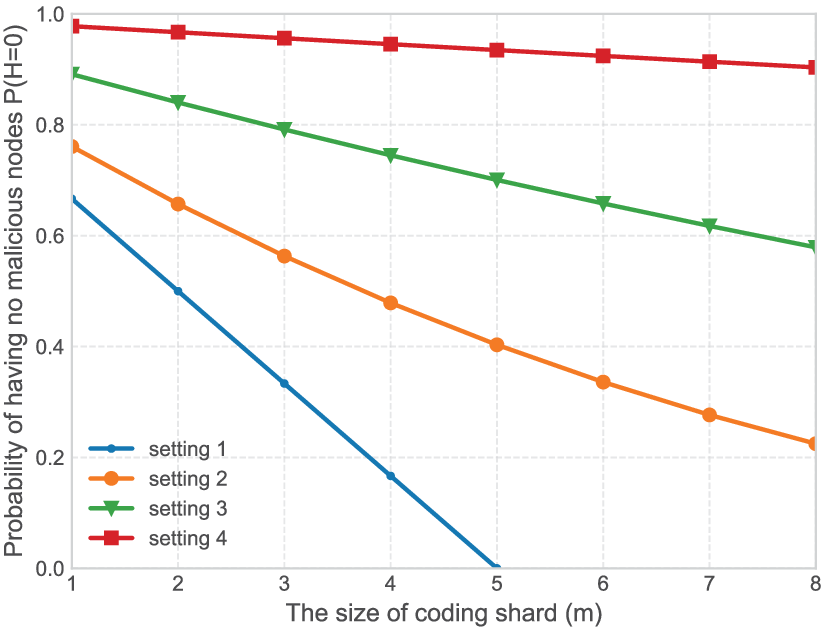}
        \label{fig:p_H:changing_m}
        \end{minipage}
    }\quad  %
    \subfigure[Probability $P(H=0)$ with different ratio $m/n$]{ %
        \begin{minipage}[t]{0.45\linewidth}
        \includegraphics[width=\linewidth]{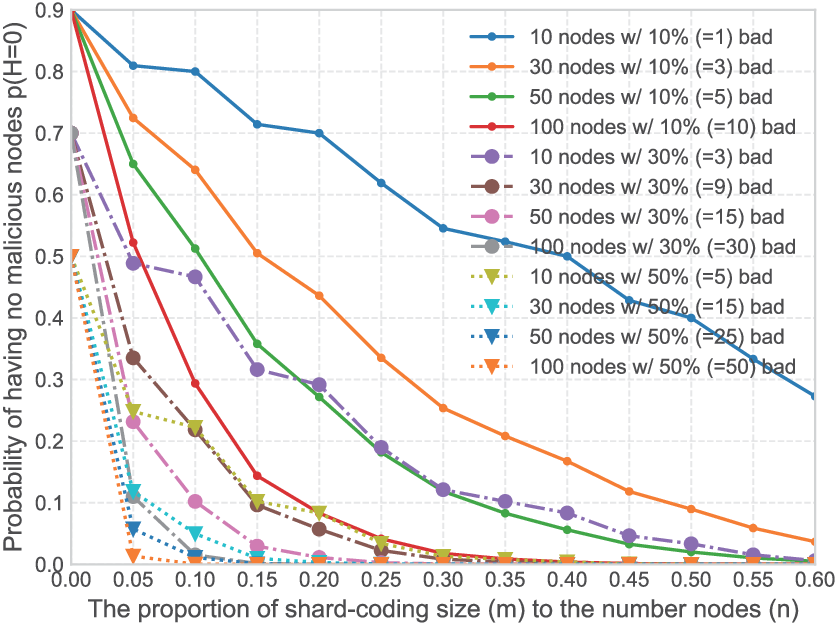} 
        \label{fig:p_H:changing_ratio}
        \end{minipage}
    }
    \caption{The influence of different settings about $m$ on the probability of having no malicious nodes in group testing $P(H=0)$.}    %
    \label{fig:p_H}    %
    \vspace{-1em}
\end{figure}

The key parameter is the shard-coding size. A larger $m$ means more divided sub-shards and thus larger members in each group. But it will also increase the computational complexity of encoded matrix construction and other computation. Therefore, we run simulations of the theoretical results of Theorem \ref{theorem:prob_no_malicious}, which is shown in Fig.~\ref{fig:p_H}. 
It is apparent from Fig.~\ref{fig:p_H:changing_m} that the successful probability decreases with the increasing of the shard-coding size.
The reason why the blue line approaches $0$ is that $n \geq (m+f+1)$ which can be derived easily from \eqnref{eqn:P_H_0}. For a reasonable trade-off between group size and computational complexity, the choice of $m$ for settings is shown in the last column of Table~\ref{table:configuration}.

Next, we investigate the impact of varying ratios ($f/n$ and $m/n$) on $P(H=0)$. The results are depicted in Fig.~\ref{fig:p_H:changing_ratio}. There is a gradual decline in the probability as $m/n$ increases. Moreover, the probability experiences a significant drop as the value of $f/n$ increases.
The result of \eqnref{eqn:P_H_0} agrees with our guess where $m$ represents the number of terms, $f$ is the numerator, and $n$ is the denominator. Since the value of each term is in the range of [0: 1], the result of cumulative multiplication will get smaller if each of them increases.

 \begin{table}[!htbp]
    \centering
    \caption{Experiment parameter settings}
    \begin{tabular}{c|c|c}
    \hline
        \textbf{Parameter} & \textbf{Value} & \textbf{Notes}  \\ \hline
        checksum value $d$ (Bytes, B) & 16 & MD5 algorithm is used here  \\ \hline
        secret key $s$ (B) & 128 & 1024 bits in general  \\ \hline
        scalar $w$ (B) & 1 &   \\ \hline
        digital signature $z$ (B) & 256 &   \\ \hline
        committee size $n$ & 1, 50, 100  & \\ \hline
        shard-coding size $m$ & $0.1n$ & \tabincell{c}{For simplicity, the ratio of \\ $m/n$ is chosen by 0.1}\\ \hline
    \end{tabular}
    \label{table:experiment_params}
\end{table}

Additionally, Fig.~\ref{fig:trials} shows the number of required group testing trials with different system parameters specified in Table~\ref{table:configuration} and we set $\rho=0.01$ empirically. 

\begin{figure}[!htb]
    \centering  %
    \subfigtopskip=2pt %
    \subfigbottomskip=2pt %
    \subfigcapskip=-5pt %
    \subfigure[Number of trials with changing $f$]{   %
        \begin{minipage}[t]{0.45\linewidth}
        \includegraphics[width=\linewidth]{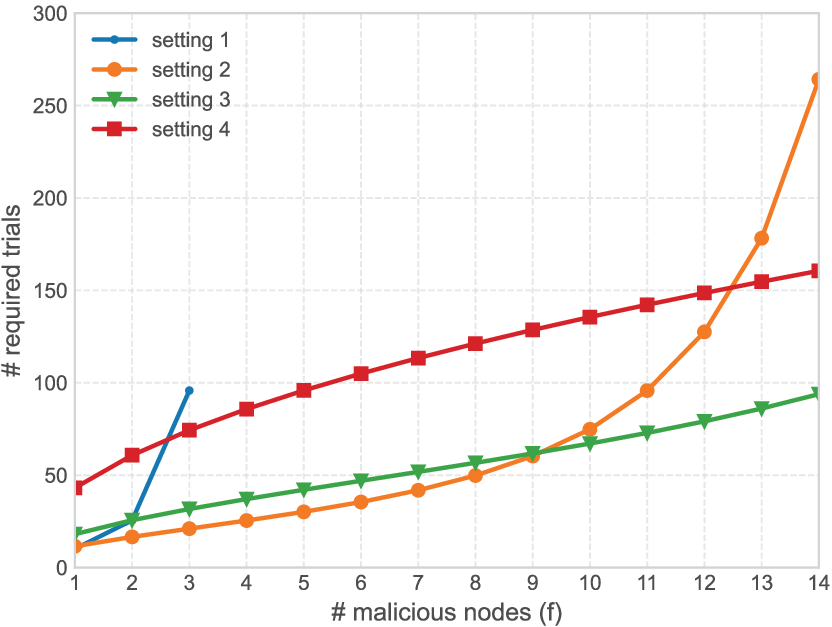}
        \label{fig:trials:a}
        \end{minipage}
    } \quad %
    \subfigure[Number of trials with changing $n$]{ %
        \begin{minipage}[t]{0.45\linewidth}
        \includegraphics[width=\linewidth]{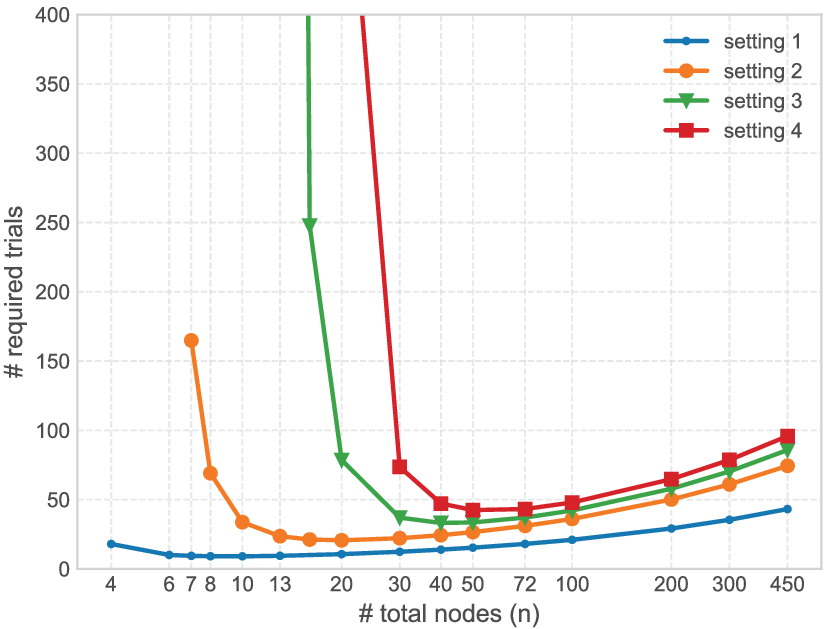} 
        \label{fig:trials:b}
        \end{minipage}
    }
    \caption{The influence of different parameters on the number of required group testing trials $T$.}    %
    \label{fig:trials}    %
    \vspace{-1.0em}
\end{figure}

In Fig.~\ref{fig:trials:a}, we show the number of trials when $f$ varies from $1$ to $14$ in different settings. It is obvious that the number of trials to identify all malicious nodes increases along with $f$. The reason why the blue line stops at $f=3$ is because $f\leq n-m-1$. 
Lines of setting 2, 3, and 4 in Fig.~\ref{fig:trials:b} do not begin at $n = 6$ because of the restriction of $n \geq (m+f+1)$. All lines begin at a high point and then show a trend from decline to rise because the ratio of $(m+f)/n$ approaches $1$ at the beginning, then $T$ reaches the minimum when the ratio approaches $1/4$.

In Fig.~\ref{fig:cost_complexity}, we numerically evaluate communication costs and computational decoding speed of RecAGT, and compare them with the other two methods. Referring to some deployments of permissioned blockchain, the specific parameters in our experiments are summarized in Table~\ref{table:experiment_params}.

\begin{figure}[htbp]
    \centering  %
    \subfigtopskip=2pt %
    \subfigbottomskip=2pt %
    \subfigcapskip=-5pt %
    \subfigure[Communication cost]{   %
        \begin{minipage}[t]{0.45\linewidth}
        \includegraphics[width=\linewidth]{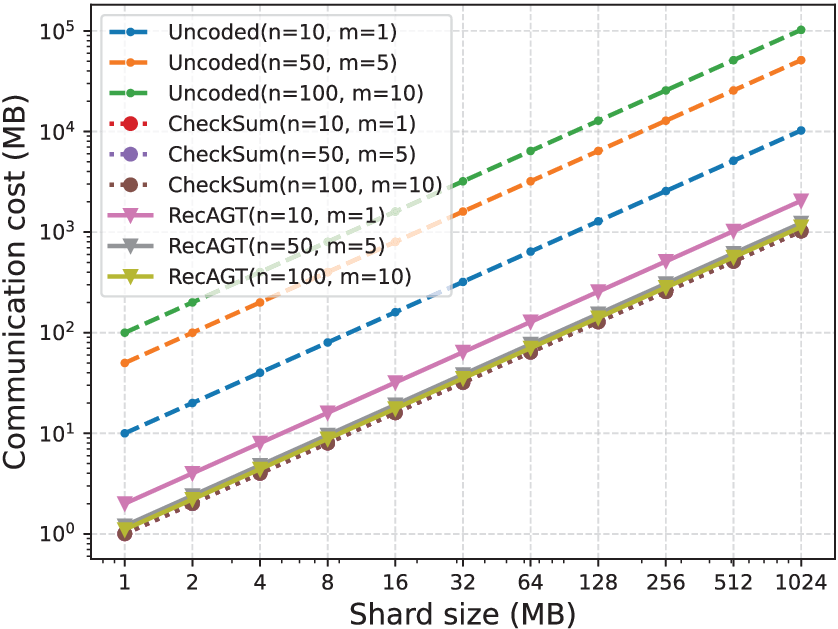}
        \label{fig:cost_complexity:cost}
        \end{minipage}
    }\quad  %
    \subfigure[Computational decoding speed (s)]{ %
        \begin{minipage}[t]{0.45\linewidth}
        \includegraphics[width=\linewidth]{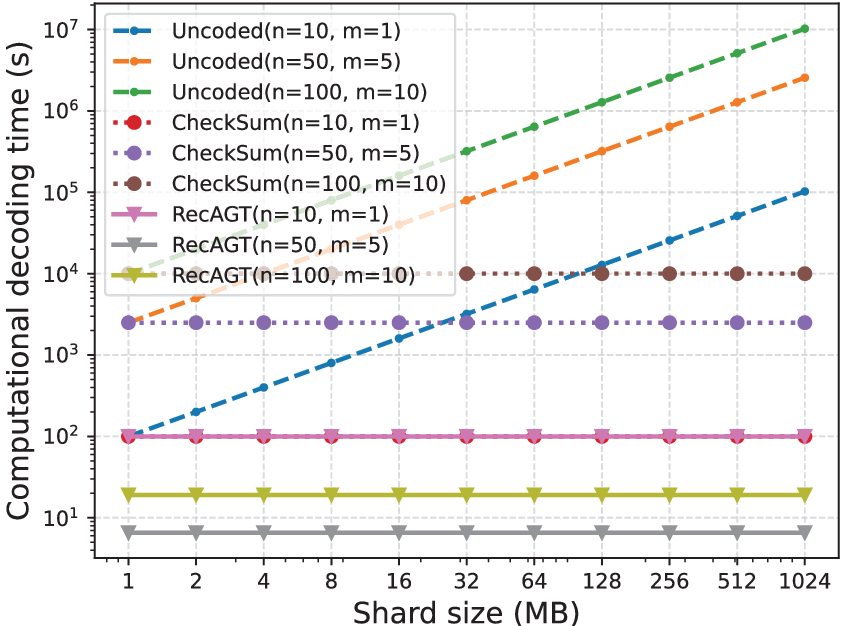} 
        \label{fig:cost_complexity:speed}
        \end{minipage}
    }
    \caption{Comparison of Uncoded, CheckSum, and RecAGT for communication cost and computational decoding speed.}    %
    \label{fig:cost_complexity}    %
\end{figure}

\textbf{Communication cost.} From Fig.~\ref{fig:cost_complexity:cost}, it becomes evident that the Uncoded scheme incurs a substantially higher communication cost than the others. This discrepancy arises from requiring all members to transmit the original shard data to the new node. It is worth mentioning that when $n$ is small, the communication cost of the RecAGT exceeds that of the CheckSum. But with the increasing $n$, the difference between them is notably small. The reason is that the CheckSum scheme needs to send the original data and checksum values of the rest, which is $(b+d\times n)$. And the decisive term in RecAGT is $[(m+1)\times b/m]$ that is explained in \eqnref{eqn:RecAGT_comu_cost}. As $n$ and $m$ increase, other factors could be ignored, so these two schemes achieve almost the same cost. However, compared with the CheckSum scheme, our RecAGT scheme encodes the original shard data and keeps the internal data of each organization confidential.

\textbf{Computational decoding speed.} In Fig.~\ref{fig:cost_complexity:speed}, we observe a significant quadratic increase in computing time for the Uncoded scheme. To illustrate, when the shard size reaches $256$ MB, the computation time exceeds an unacceptable $10,000$ s. The lines of CheckSum and RecAGT appear constant because they do not directly manipulate shard data. CheckSum uses check codes to validate the consistency of data received from other members. The RecAGT scheme uses the inverse of the Vandermonde matrix to identify perturbed data. Our adoption of the polynomial interpolation method demonstrates superior speed and efficiency compared to the other two schemes, as evidenced by our experimental results.

\section{Conclusion}
\label{section:conclusion}
In this paper, we propose RecAGT scheme for the identification of potential malicious nodes, which focuses on reducing communication overhead and identifying potential malicious nodes. Specifically, we design the shard testable codes to encode original data. And we come up with an identity proof using the digital signature and choose an adaptive group testing method to make the required number of trials as small as possible. The simulation results demonstrate that our proposed RecAGT scheme can efficiently identify malicious nodes and reduce communication and computational costs.
\subsubsection{Acknowledgements}
This work was supported in part by the National Natural Science Foundation of China (62072321, 61972272), the Six Talent Peak Project of Jiangsu Province (XYDXX-084), the China Postdoctoral Science Foundation (2020M671597), the Jiangsu Postdoctoral Research Foundation (2020Z100), Suzhou Planning Project of Science and Technology (SS202023), the Future Network Scientific Research Fund Project (FNSRFP-2021-YB-38), Natural Science Foundation of the Higher Education Institutions of Jiangsu Province (22KJA520007, 20KJB520002), the Collaborative Innovation Center of Novel Software Technology and Industrialization, and Soochow University Interdisciplinary Research Project for Young Scholars in the Humanities.

\bibliographystyle{splncs04}
\bibliography{documents}

\end{document}